\documentclass[10pt,doublecolumn,twoside,journal]{IEEEtran}
\IEEEoverridecommandlockouts
\usepackage{CJK}
\usepackage{subfigure}
\usepackage{colortbl}
\usepackage{bm}
\usepackage{graphicx}  
\usepackage{url}       

\usepackage{amsmath}
\usepackage{cite}

\usepackage{amsfonts,amssymb}
\usepackage{comment} 
\usepackage{stfloats}

\usepackage{cases}
\usepackage{algorithm}
\usepackage{multirow}
\usepackage{algorithmic}
\usepackage{float}
\usepackage{float}

\usepackage{graphicx}  
\usepackage{float}  
\usepackage{subfigure}  

\newtheorem{theorem}{{Theorem}}
\newtheorem{lemma}{{Lemma}}

\newtheorem{proof}{Proof}[section]

\newcommand{\ls}[1]
    {\dimen0=\fontdimen6\the\font
     \lineskip=#1\dimen0
     \advance\lineskip.5\fontdimen5\the\font
     \advance\lineskip-\dimen0
     \lineskiplimit=.9\lineskip
     \baselineskip=\lineskip
     \advance\baselineskip\dimen0
     \normallineskip\lineskip
     \normallineskiplimit\lineskiplimit
     \normalbaselineskip\baselineskip
     \ignorespaces
    }

\begin{document}
%
\title{Multi-Frequency Resonating Based Magnetic Induction Underground Emergency Communications with Diverse Mediums}
\author{Jianyu~Wang,~\IEEEmembership{Member,~IEEE,}
        Zhichao~Li,
        Wenchi~Cheng,~\IEEEmembership{Senior~Member,~IEEE,}
        Wei~Zhang,~\IEEEmembership{Fellow,~IEEE,}
        and~Hailin~Zhang,~\IEEEmembership{Member,~IEEE}
\thanks{

Jianyu Wang, Zhichao Li, Wenchi Cheng, and Hailin Zhang are with State Key Laboratory of Integrated Services Networks, Xidian University, Xi'an, 710071, China (e-mails: wangjianyu@xidian.edu.cn; 24011210926@stu.xidian.edu.cn; wccheng@xidian.edu.cn; hlzhang@xidian.edu.cn)

Wei Zhang is with School of Electrical Engineering and Telecommunications, the University of New South Wales, Sydney, Australia (e-mail: w.zhang@unsw.edu.au).
}
}

\maketitle
\thispagestyle{empty}
\begin{abstract}
Magnetic induction (MI) communication is an effective underground emergency communication technique after disasters such as landslides, mine collapses, and earthquakes, due to its advantages in mediums such as soil, concrete, and metals.
However, the propagation mediums in practical MI based underground emergency communications are usually diverse and composed randomly due to the impact of disasters, which poses a challenge for MI communication in practical applications.
In this paper, we formulate a statistical fading channel model, which reflects the random composition of diverse mediums and is shown to follow a lognormal distribution.
To mitigate the impact of diverse medium fading, Multi-frequency Resonating Compensation (MuReC) based coils are used to achieve multi-band transmission.
Then, we analyze the performance of MuReC based multi-band MI communication with diverse medium fading and derive the expressions of signal-to-noise ratio (SNR) probability density functions, ergodic capaciteis, average bit error rates (BERs), and outage probabilities for both multiplexing and diversity cases.
Numerical results show that MuReC based multi-band transmission schemes can effectively reduce the impact of diverse medium fading and enhance the performance.
\end{abstract}


\section{introduction}
\IEEEPARstart{D}{ue} to the advantages in special mediums such as soil, concrete, and metals, magnetic induction (MI) communications have gained substantial attention in recent years and are considered to support novel and important applications in future underground emergency rescues~\cite{Guanghua_Liu_Magazine,Xin_Tan_APM}. Wireless channels in MI communications can be classified into the homogeneous-medium channel, where the propagation path of the magnetic field from the transmit antenna to the receive antenna contains only one medium, and inhomogeneous-medium channels, where the propagation path of the magnetic field from the transmit antenna to the receive antenna contains more than one medium.

Many works related to MI communications have been done based on homogeneous-medium channels and demonstrated performance gains.
For examples, various coil antenna designs, such as circular loop antennas with proper partial overlapping~\cite{Access_8} and heterogeneous multipole loop antenna array~\cite{Kim_2016}, have been studied with air as the medium to mitigate crosstalks in MI based multiple-input multiple-output (MI-MIMO).
The channel modeling~\cite{Hongzhi_Guo_TAP}, practical design~\cite{Practical_M2I, Practical_M2I_OPtimal}, and full-duplex (FD) relaying design~\cite{Hongzhi_Guo_INFOCOM} for metamaterial enhanced magnetic induction ($\text{M}^2\text{I}$) system have been investigated based on single lossy medium.
\begin{figure*}
\centering
\includegraphics[width=0.95\textwidth]{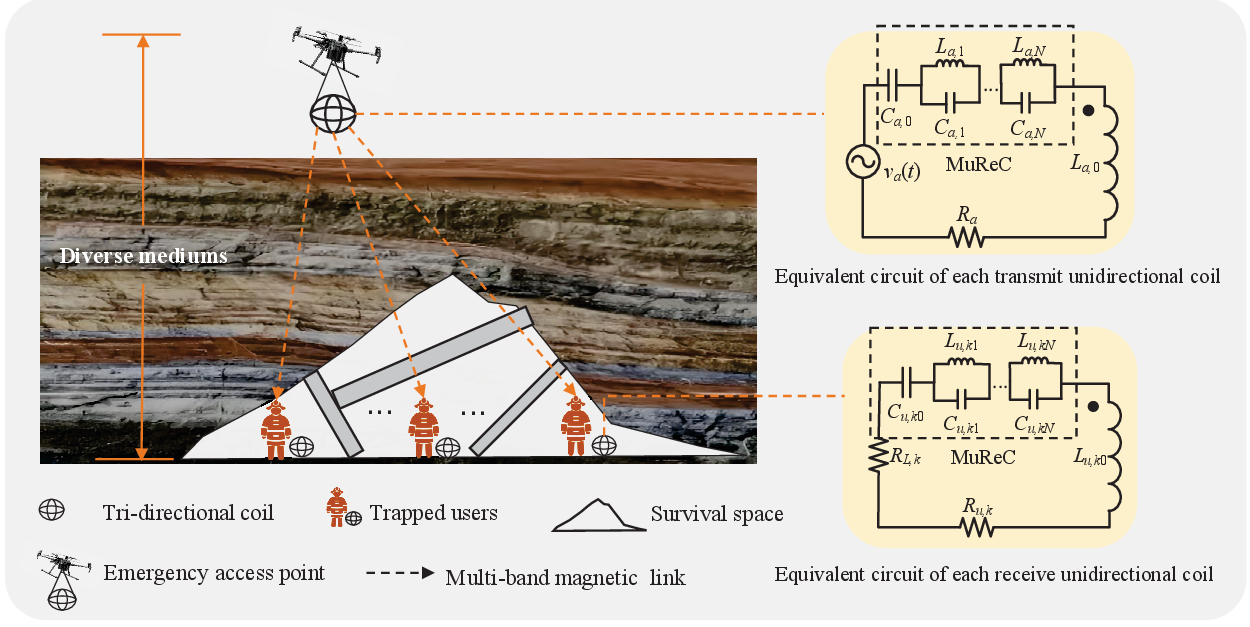}
\caption{Magnetic Induction based multi-band emergency communication with diverse mediums.} \label{fig:Scenario}
\end{figure*}
In addition, the authors of~\cite{Gang_Yang_TSP} proposed a magnetic beamforming scheme for MI-MIMO system, which uses air as a medium, and introduced a low-complexity channel estimation method based on least squares.
To address magnetic channel estimation errors, robust beamforming technique has been explored in~\cite{Yixin_Zhang_TSP} and~\cite{Robust_Beamforming_2}. To address the scheduling challenge in multi-node underground transmission, the authors of~\cite{Wireless_Charge_Steven_Kisseleff} proposed a combined beamforming and scheduling strategy, where the underground medium is soil and assumed to be homogeneous.
Our previous studies demonstrated that integrating magnetic beamforming with multi-frequency resonating compensation (MuReC)~\cite{JSAC_2022}, orthogonal frequency division multiplexing (OFDM)~\cite{ICCC_2021}, and backscatter based FD~\cite{Backscatter_FD} can substantially enhance the achievable data rate and energy efficiency for MI communications in homogeneous-medium channels.
However, works mentioned above do not consider or evaluate the impact of multiple propagation mediums on performances. In practical post-disaster scenarios, the propagation mediums of MI based emergency communications are usually inhomogeneous and contain more than one medium.

Considering the inhomogeneity of surrounding environments, the authors of~\cite{Joint_Channel_Antenna_Modeling} developed an exact full-wave model and a near field approximated model for the MI channel. In~\cite{Hongzhi_Guo_Metal}, MI is used to provide a reliable and flexible solution for wireless sensor and robotic networks in metal-constrained environments, where the propagation channel contains one metal layer and two other layers. Taking into account the impact of the boundary between water and air, the modulation scheme and channel capacity are studied in~\cite{Border_Modulation} and~\cite{Border_Capacity}, respectively.
Based on skin depths of multiple mediums, the authors of~\cite{Digital_MI_Steven} gave the expression for the inhomogeneous underground MI channel.
Although numerous studies have explored the effects of multiple propagation mediums on MI communications, existing research typically assumes predetermined and known medium dimensions. However, in practical underground emergency rescues, the propagation mediums are not only diverse but also randomly composed due to disasters, which presents a challenge for the practical implementation of MI communication systems in emergency rescues.

To the best of our knowledge, the impact of diverse and randomly composed mediums on the performances of MI communications is currently unknown and how to deal with the fading caused by diverse and irregular mediums was not discussed. In this paper we propose a statistical channel model for MI based underground emergency communications, that accounts for the random composition of diverse mediums caused by disasters. The proposed model is demonstrated to follow a lognormal distribution. To address performance degradation caused by such fading, we employ MuReC based tri-directional (TD) coils to achieve simultaneous transmission across multiple bands. Since simultaneous deep fading across all bands occurs with low probability, the proposed multi-band approach demonstrates significant advantages over conventional single-band MI systems, which results in improved system performance metrics, including higher ergodic capacity, lower average bit error rate (BER), and reduced network outage probability. We present comprehensive theoretical analysis for MuReC based multi-band MI system under diverse medium fading and derive the expressions for signal-to-noise ratio (SNR) probability density functions, ergodic channel capacities, average BERs, and network outage probabilities for both frequency multiplexing and frequency diversity cases.

The remainder of this paper is structured as follows. In Section~II, we formulate the mathematical model for MuReC based multi-band underground emergency communication network. In Section~III, we formulate a statistical diverse medium fading channel model, which reflects the random composition of diverse transmission mediums and is shown to follow a lognormal distribution. In Section~IV, theoretical analysis for MuReC based multi-band system under diverse medium fading is presented and the expressions for SNR probability density functions, ergodic channel capacities, average BERs, and network outage probabilities for both frequency multiplexing and frequency diversity cases are derived. We conduct numerical evaluations in Section~V and conclude in Section~VI.
\section{System Model}
In this section, the mathematical model for MuReC based multi-band underground emergency communication network is formulated.

We consider an MI based emergency communication as shown Fig.~\ref{fig:Scenario}, where an emergency access point (EAP), such as a drone, transmits information to $K+1$ underground trapped users (TUs) with $N+1$ bands, denoted by $\mathcal{B}=\left\{\mathcal{B}_{0},...,\mathcal{B}_{n},...,\mathcal{B}_{N}\right\}$. We denote by $\mathcal{F}=\{f_{0},...,f_{n},...,f_{N}\}$ the center frequencies of the $N+1$ bands, where the center frequency of $\mathcal{B}_{n}$ is $f_{n}$. Then, $\mathcal{B}_{n}$ can be expressed as $\mathcal{B}_{n}=\{f|f_{n}-B_{n}/2\leq f\leq f_{n}+B_{n}/2\}$, where $f$ denotes the frequency and $B_{n}$ is the bandwidth of $\mathcal{B}_{n}$.
TUs access the network in time division multiple access (TDMA) manner, such that all the bands can be used for the transmission corresponding to each user to mitigate the impact of diverse medium fading. The frame duration, which is normalized to 1, is equally divided into $K+1$ time-slots and the $k$-th time-slot is assigned to the $k$-th TU. We denote by $\mathcal T_k=\{t|k/(K+1)\leq t\leq (k+1)/(K+1)\}$ the $k$-th time-slot, where $t$ denotes the time ($k\in\{0,1,...,K\}$). The underground mediums are inhomogeneous and consist of diverse subtances. To guarantee the omnidirectional radiation and reception of MI signals, the EAP and TUs are equipped with tri-directional (TD) coils, which consist of three mutually perpendicular unidirectional coils~\cite{TD_Coil_TVT,TD_Coil_ICC}.

Based on our previous work in~\cite{JSAC_2022}, we employ MuReC circuit for the unidirectional coil to make it resonate at multiple frequencies. With MuReC, both the TD coils of the EAP and TUs can simultaneously resonate at $\mathcal{F}$, such that the receive power corresponding to each band can be maximized.
The equivalent circuit corresponding to each unidirectional coil of the EAP is shown in the upper right corner of Fig.~\ref{fig:Scenario}, where $v_a(t)$, $L_{a,0}$, and $R_a$ are time-domain input voltage signal, coil self-inductance and the coil self-resistance, respectively. Similarly, the equivalent circuit corresponding to each unidirectional coil of the $k$-th TU is shown in the lower right corner of Fig.~\ref{fig:Scenario}, where $L_{u,k0}$, $R_{u,k}$, and $R_{L,k}$ are coil self-inductance, coil self-resistance, and load resistance, respectively. The parameter values of MuReC circuits, that is, $\{C_{a,0}$, ..., $C_{a,N}\}$, $\{L_{a,1}$, ..., $L_{a,N}\}$, $\{C_{u,k0}$, ..., $C_{u,kN}\}$, and $\{L_{u,k1}$, ..., $L_{u,kN}\}$ can be configured according to~\cite{JSAC_2022} to make the unidirectional coil of the EAP and TUs simultaneously resonate at $\mathcal{F}$.
Fig.~\ref{fig:MuReC_Coil} shows an example of the impendence characteristic of the MuReC coil, which has five resonant frequencies and there is a pole between two adjacent resonant frequencies.
Signals near resonant frequencies in MuReC coil are with low impendence loss. Also, due to the poles between two adjacent resonant frequencies, the multi-band signals in the MuReC coil can be separated~\cite{JSAC_2022}.
\begin{figure}
\centering
\includegraphics[width=0.5\textwidth]{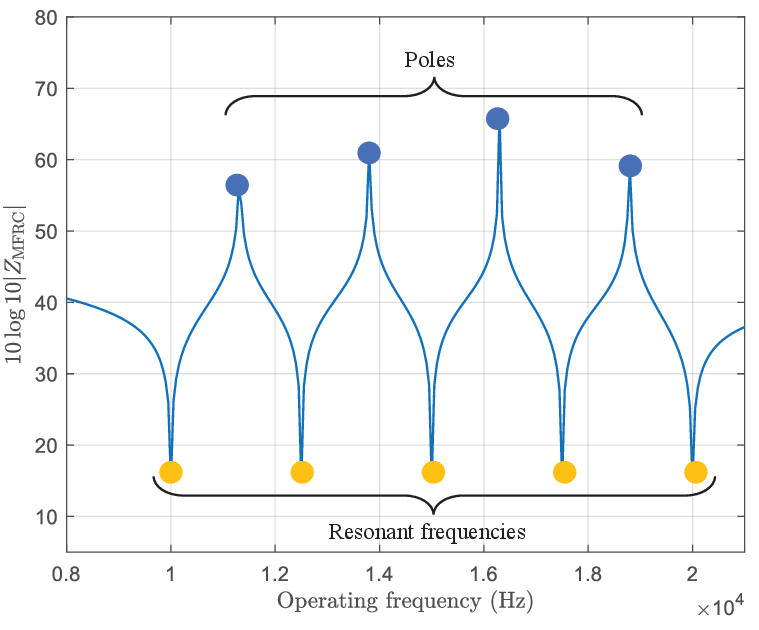}
\caption{Impendence characteristic of MuReC coil.} \label{fig:MuReC_Coil}
\end{figure}

We denote by $Z_a(f)$ and $Z_{u,k}(f)$ the impendence of each unidirectional coil of the EAP and the $k$-th TU, respectively, which are given by $Z_a(f)\!=\!R_a\!+\!j\left[\!\frac{4\pi^2f^2 L_{a,0} C_{a,0}-1}{2\pi f C_{a,0}}\!+\!\sum_{n=1}^N \frac{2\pi f L_{a,n}}{1-4\pi^2f^2 L_{a,n} C_{a,n}}\!\right]$
and $Z_{u,k}(f)\!=\!R_{u,k}\!\!+\!\!R_{L,k}\!\!+\!\!j\left[\frac{4\pi^2f^2 L_{u,k0} C_{u,k0}\!-\!1}{2\pi f C_{u,k0}}\!+\!\sum_{n\!=\!1}^{N} \frac{2\pi f L_{u,kn}}{1-4\pi^2f^2 L_{u,kn} C_{u,kn}}\right]$, respectively.
The impendence matrix of the EAP is denoted by $\mathbf{Z}_{a}(f)$, where diagonal and off-diagonal entries are the coil self-resistances and the mutual inductances among the coils of EAP, respectively. Since the three unidirectional coils of the TD coil are perpendicular to each other, their mutual inductances are zero. Thus, the expression of $\mathbf{Z}_{a}(f)$ can be given by $\mathbf{Z}_{a}(f)=Z_a(f)\mathbf{E}$, where $\mathbf{E}$ is the identity matrix. Similarly, we denote by $\mathbf{Z}_{u,k}(f)$ the impendence matrix of the $k$-th TU, which can be given by $\mathbf{Z}_{u,k}(f)=Z_{u,k}(f)\mathbf{E}$.
The mutual inductance matrix between the EAP and the $k$-th TU at frequency $f$, denoted by $\mathbf{M}_{k}(f)$, can be written as follows:
\begin{equation}\label{eq:M_k}
\mathbf{M}_{k}(f)=\left[\begin{array}{cccc}
M_{k,11}(f) & M_{k,12}(f) &  M_{k, 13}(f) \\
M_{k,21}(f) & M_{k,22}(f) &  M_{k, 23}(f) \\
M_{k,31}(f) & M_{k,32}(f) &  M_{k, 33}(f)
\end{array}\right],
\end{equation}
where $M_{k,pq}(f)$ is the mutual inductance between the $p$-th unidirectional coil of EAP and the $q$-th unidirectional coil of the $k$-th TU ($p,q\in\{1,2,3\}$).
We denote the mutual inductance matrix between the $v$-th TU and the $w$-th TU at frequency $f$ by $\mathbf{\tilde M}_{vw}(f)$, which can be given by
\begin{equation}\label{eq:M_k_vw}
\mathbf{\tilde M}_{vw}(f)=\left[\begin{array}{cccc}
\tilde M_{vw,11}(f) & \tilde M_{vw,12}(f) &  \tilde M_{vw, 13}(f) \\
\tilde M_{vw,21}(f) & \tilde M_{vw,22}(f) &  \tilde M_{vw, 23}(f) \\
\tilde M_{vw,31}(f) & \tilde M_{vw,32}(f) &  \tilde M_{vw, 33}(f)
\end{array}\right],
\end{equation}
where $\tilde M_{vw,pq}(f)$ is the mutual inductance between the $p$-th coil of the $v$-th TU and the $q$-th coil of the $w$-th TU.

Let $x_{k,n}(t)$ denote the time-limited transmit signal corresponding to the $k$-th time-slot and the band $\mathcal{B}_{n}$, whose Fourier transform is $X_{k,n}(f)$. The power of $x_{k,n}(t)$ on a unit resistance, denoted by $P_{k,n}$, can be given by $P_{k,n}=(K+1)\int_{\mathcal T_k}|x_{k,n}(t)|^2dt$.
Most energy of $X_{k,n}(f)$ is concentrated within $\mathcal{B}_{n}$ and satisfies $\int_{\mathcal{B}_{n}}|X_{k,n}(f)|^2df\!\approx\!\int_{\mathcal T_k}|x_{k,n}(t)|^2dt\!=\!P_{k,n}/(K\!+\!1)$.
Then, $x_{k,n}(t)$ is mapped to the coil input voltage of the EAP, which results in the expression of $v_a(t)$ as $v_a(t)=\sum_{k=0}^{K}\sum_{n=0}^{N}x_{k,n}(t)$,
whose Fourier transform is
\begin{equation}\label{eq:V_tf}
V_a(f)=\sum_{k=0}^{K}\sum_{n=0}^{N}X_{k,n}(f).
\end{equation}
The EAP loads the same transmit signal on three unidirectional coils, such that the omnidirectional transmission can be guaranteed~\cite{TD_Coil_TVT,TD_Coil_ICC}. Based on~\eqref{eq:V_tf}, we can write the voltage vector of EAP in the frequency-domain as $\mathbf{V}_a(f)
\!=\![V_a(f),V_a(f),V_a(f)]^T$.
We denote by $\mathbf{I}_{a}(f)\!=\!\left[{I}_{a,1}(f),{I}_{a,2}(f),{I}_{a,3}(f)\right]^T$ and $\mathbf{I}_{k}(f)\!=\!\left[{I}_{k,1}(f),{I}_{k,2}(f),{I}_{k,3}(f)\right]^T$ the current vectors of the EAP and the $k$-th TU, respectively.
According to Faraday's electromagnetic induction law and Kirchhoff's voltage law, we have
\begin{equation}\label{eq:Ikf}
\mathbf{I}_{k}(f)\!\!=\!\!\mathbf{Z}_{u,k}(f)^{\!-\!1}\!\Big[\!j2\pi f\mathbf{M}_{k}(f)\mathbf{I}_{a}(f)\!-\!\!\!\sum_{v\ne k}^{K}\!j2\pi f\mathbf{\tilde M}_{vk}(f)\mathbf{I}_{v}(f)\!\Big]
\end{equation}
and
\begin{equation}\label{eq:Vkf}
\mathbf{V}_a(f)=\mathbf{Z}_a(f)\mathbf{I}_a(f)-\sum_{k=0}^{K}j2\pi f\mathbf{M}_k(f)^T\mathbf{I}_k(f).
\end{equation}
In~\eqref{eq:Ikf}, $\sum_{v\ne k}^{K}\!j2\pi f\mathbf{\tilde M}_{vk}(f)\mathbf{I}_{v}(f)$ is the secondary induced voltage generated by the receive currents of other TUs, which is much smaller than the primary induced voltage $j2\pi f\mathbf{M}_{k}(f)\mathbf{I}_{a}(f)$ and can be ignored~\cite{Gang_Yang_TSP,Zhangyu_Li_GLOBECOM,Yixin_Zhang_TSP}.
Therefore, \eqref{eq:Ikf} can be re-written as follows:
\begin{equation}\label{eq:Ikf-rw}
\mathbf{I}_{k}(f)=\frac{j2\pi f\mathbf{M}_{k}(f)\mathbf{I}_{a}(f)}{Z_{u,k}(f)}.
\end{equation}
Plugging~\eqref{eq:Ikf-rw} into \eqref{eq:Vkf}, we have
\begin{equation}\label{eq:Vkf-rw}
\begin{split}
\mathbf{V}_a(f)&=\left[\mathbf{Z}_a(f)+\sum_{k=0}^{K}\frac{4\pi^2 f^2\mathbf{M}_{k}(f)^T\mathbf{M}_k(f)}{Z_{u,k}(f)}\right]\mathbf{I}_{a}(f)
\\& \mathop\approx^{\text{(a)}}\mathbf{Z}_a(f)\mathbf{I}_{a}(f),
\end{split}
\end{equation}
where (a) is because the transmission is long, such that $\mathbf{Z}_a(f)\gg\sum_{k=0}^{K}\frac{4\pi^2 f^2\mathbf{M}_{k}(f)^T\mathbf{M}_k(f)}{Z_{u,k}(f)}$~\cite{TD_Coil_TVT,Hongzhi_Guo_TAP,Active_Relaying}. Combining \eqref{eq:Vkf-rw} and \eqref{eq:Ikf-rw}, the receive signals of the $k$-th TU, which are the voltages of load resistors, denoted by $\mathbf{Y}_{k}(f)$, can be given by
\begin{equation}\label{eq:ykf}
\mathbf{Y}_{k}(f)=\frac{j2\pi f\mathbf{M}_{k}(f)\mathbf{V}_{a}(f)R_{L,k}}{Z_{u,k}(f)Z_{a}(f)}+\mathbf{W}_{k}(f),
\end{equation}
where $\mathbf{W}_{k}(f)$ is the noise vector corresponding to the $k$-th TU. To guarantee the omnidirectional reception of MI signals, equal gain combining is used for the TD coil of the $k$-th TU~\cite{TD_Coil_TVT,TD_Coil_ICC}.
Then, we can write the receive signal of the $k$-th TU after combining, denoted by $Y_{k}(f)$, as follows:
\begin{equation}\label{eq:ykfac}
Y_{k}(f)=H_k(f)V_a(f)+W_{k}(f),
\end{equation}
where
\begin{equation}\label{eq:hkf}
H_k(f)=\frac{j2\pi f \mathbf{o}\mathbf{M}_{k}(f)\mathbf{o}^TR_{L,k}}{Z_{u,k}(f)Z_{a}(f)},
\end{equation}
${W}_{k}(f)\!=\!\mathbf{o}\mathbf{W}_{k}(f)$, and $\mathbf{o}=[1,1,1]/\sqrt{3}$. We assume that $\mathcal{B}_{n}$ is narrow such that $H_k(f)\approx H_k(f_{n})$ for $ f\in \mathcal{B}_{n}$~\cite{Yixin_Zhang_TSP,JSAC_2022,Joint_Channel_Antenna_Modeling}.
Then, the component of \eqref{eq:ykfac} corresponding to $\mathcal{B}_{n}$, denoted by
$Y_{k,n}(f)$, can be obtained with band filtering for $\mathcal B_n$ as follows:
\begin{equation}\label{eq:ykfac_n_rw}
\begin{split}
Y_{k,n}(f)=H_k(f_{n})\sum_{v=0}^{K}\!X_{v,n}(f)\!+\!W_{k}(f),f\in\mathcal{B}_{n}.
\end{split}
\end{equation}
Let $y_{k,n}(t)$ denote the inverse Fourier transform for $Y_{k,n}(f)$. Then, in the $k$-th time-slot $\mathcal{T}_k$, $y_{k,n}(t)$ can be given by
\begin{equation}\label{eq:ykfac_n}
\begin{split}
y_{k,n}(t)=H_k(f_{n})x_{k,n}(t)+w_{k,n}(t),t\in\mathcal{T}_{k},
\end{split}
\end{equation}
where $w_{k,n}(t)$ is the noise signal corresponding to the $n$-th band of the $k$-th TU. For any time $t$, we assume that $w_{k,n}(t)$ is white and follows $\mathcal{CN}(0,N_0B_n)$, where $N_0$ is the noise power spectral density.

\textit{Remarks on (\ref{eq:ykfac_n}):} Due to MuReC circuits, each unidirectional coil of the EAP and TUs resonates at $f_n$. Thus, the term $H_k(f_n)$ in \eqref{eq:ykfac_n} equals to $\frac{j2\pi f \mathbf{o}\mathbf{M}_{k}(f)\mathbf{o}^TR_{L,k}}{(R_{u,k}+R_{L,k})R_{a}}$, whose module is the maximum value of $|H_k(f)|$. Therefore, as compared with conventional coils without MuReC, the MuReC based TD coils in this paper can guarantee efficient radiation and reception of MI signals in multiple bands.

Based on \eqref{eq:Vkf-rw}, we can derive the transmit power of the EAP, denoted by $P_t$, as follows:
\begin{equation}\label{eq:P_tn}
\begin{split}
P_{t}&\!=\!\int_{\mathcal{B}}\frac{1}{2}\mathbb{R}\left({\mathbf{ V}_{a}(f)^H\mathbf{ I}_{a}(f)}\right)df\
\\&\!=\!\int_{\mathcal{B}}\mathbb{R}\left(\frac{\mathbf{ V}_{a}(f)^H\mathbf{ V}_{a}(f)}{2Z_a(f)}\right)df
\\&\!=\!\sum_{k=0}^{K}\sum_{n=0}^{N}\!\int_{\mathcal{B}_{n}}\!\!\!\!\frac{3|X_{k,n}(f)|^2R_a}{2|Z_a(f)|^2}df\!\mathop\leq^{\text{(a)}}\!\frac{3}{2R_a(K\!+\!1)}\!\sum_{k=0}^{K}\!\sum_{n=0}^{N}P_{k,n},
\end{split}
\end{equation}
where (a) results from $R_a^2\leq|Z_a(f)|^2$.
\section{Diverse Medium Fading Model}
In this section, we formulate a statistical diverse medium fading channel model, which reflects the random composition of diverse transmission mediums and is shown to follow a lognormal distribution.

In most underground scenarios where the transmission mediums do not contain magnetite, $M_{k,pq}(f)$ can be expressed as follows~\cite{Zhi_Sun_Underground,Digital_MI_Steven,Zhi_Sun_2013}:
\begin{equation}\label{eq:M_k_eq}
M_{k,pq}(f)=\mu \pi N_tN_r \frac{a_t^2a_r^2}{4 d_k^3}J_{k,pq}G_k(f),
\end{equation}
where $\mu$ is the permeability, $a_t$ is the coil radius of EAP, $a_r$ is the coil radius of TU, $N_t$ is the number of turns of EAP's coils,  $N_r$ is the number of turns of TUs' coils, $d_k$ is the distance between the EAP and the $k$-th TU, $J_{k,pq}$ is the angular misalignment factor between the $p$-th unidirectional coil of the EAP and the $q$-th unidirectional coil of the $k$-th TU~\cite{Steven_Interference_Polarization}, and $G_k(f)$ is the additional loss factor due to the eddy current. Substituting \eqref{eq:M_k_eq} into \eqref{eq:M_k}, we have
\begin{equation}\label{eq:M_k_rw}
\mathbf{M}_{k}(f)=\mathbf{\bar M}_kG_k(f),
\end{equation}
where
\begin{equation}\label{eq:bar_M_k_rw}
\mathbf{\bar M}_{k}=\left[\begin{array}{cccc}
\bar M_{k,11} & \bar M_{k,12} &  \bar M_{k, 13} \\
\bar M_{k,21}& \bar M_{k,22} &  \bar M_{k, 23}\\
\bar M_{k,31} & \bar M_{k,32} &  \bar M_{k, 33}
\end{array}\right],
\end{equation}
with $\bar{M}_{k,pq}=\mu \pi N_tN_r a_t^2a_r^2J_{k,pq}/{4 d_k^3}$. In fact, the orientation of the coil, which is denoted by $J_{k,pq}$ in \eqref{eq:M_k_eq}, can also cause fading for the MI channel. This kind of fading has been thoroughly studied in~\cite{TD_Coil_TVT,TD_Coil_ICC,Active_Relaying} and can be significantly reduced with TD coils. Thus, in this paper we mainly consider the fading caused by diverse mediums, which will be discussed in the following.
Suppose there are $W_k$ kinds of mediums on the transmission path from the EAP to the $k$-th TU.
The skin depth and the length along the transmission path of the $i$-th medium are denoted by $\delta_{k,i}(f)$ and $\Delta r_{k,i}$, respectively.
In practical post-disaster underground scenarios, $\Delta r_{k,i}$ is irregular. We model $\Delta r_{k,i}$ as a random variable with a mean of $E_{k,i}$ and a variance of $D_{k,i}$. Then, we have the following lemma for $G_k(f)$.
\begin{lemma}
If $W_k\to +\infty$, $G_k(f)$ follows a lognormal distribution as follows:
\begin{equation}\label{eq:G_kf_distribution}
\log_e\left[G_k(f)\right]\sim\mathcal N\left(\bar{E}_k(f),\bar{D}_k(f)\right),
\end{equation}
where
\begin{equation}\label{eq:bar_E}
\bar{E}_k(f)=\sum_i^{W_k} \frac{-E_{k,i}}{\delta_{k,i}(f)},
\end{equation}
and
\begin{equation}\label{eq:bar_D}
\bar{D}_k(f)=\sum_i^{W_k} \frac{D_{k,i}}{\delta_{k,i}^2(f)}.
\end{equation}
\end{lemma}
\begin{proof}
According to~\cite{Digital_MI_Steven}, the expression of $G_k(f)$ can be given by
\begin{equation}\label{eq:G_kf}
G_k(f)=\prod_i^{W_k} \exp\left[{-\frac{\Delta r_{k,i}}{\delta_{k,i}(f)}}\right]=\exp\left[{-\sum_i^{W_k} \frac{\Delta r_{k,i}}{\delta_{k,i}(f)}}\right].
\end{equation}
In accordance with the central limit theorem, no matter what distribution $\Delta r_{k,i}$ follows, if $W_k\to+\infty$, ${-\sum_i^{W_k} \frac{\Delta r_{k,i}}{\delta_{k,i}(f)}}$ follows $\mathcal N\left(\bar{E}_k(f),\bar{D}_k(f)\right)$. Therefore, $G_k(f)$ follows a lognormal distribution as~\eqref{eq:G_kf_distribution}. Lemma 1 follows.
\end{proof}

\textit{Remarks on Lemma 1:} Lemma 1 reveals that if underground mediums are very diverse and irregular, the eddy current attenuation of the magnetic field follows a lognormal distribution, which lays the foundation for formulating the statistical model of the MI channel in the underground buried environment after disasters. We define $\widetilde H_k(f)\!=\!\log_e\!\!\left[\left|\frac{j2\pi fR_{L,k}\mathbf{o}\mathbf{\bar M}_{k}\mathbf{o}^T}{Z_{u,k}(f)Z_{a}(f)}\right|\right]$, $\theta_k(f)\!=\!\arg\!\left[\frac{j2\pi fR_{L,k}\mathbf{o}\mathbf{\bar M}_{k}\mathbf{o}^T}{Z_{u,k}(f)Z_{a}(f)}\right]$, and $\bar H_k(f)\!=\!\log_e\!\!\left[\frac{4\pi^2 fR_{L,k}^2(\mathbf{o}\mathbf{\bar M}_{k}\mathbf{o}^T)^2}{|Z_{u,k}(f)|^2|Z_{a}(f)|^2}\right]$.
Based on Lemma~1, we have the following Theorem for $H_k(f)$ and $|H_k(f)|^2$.
\begin{theorem}
The module of $H_k(f)$ follows a lognormal distribution as
\begin{equation}\label{eq:H_k_module_distribution}
\log_e\left[ |H_k(f)|\right]\sim\mathcal N\left(\widetilde H_k(f)+\bar{E}_k(f),\bar{D}_k(f)\right),
\end{equation}
the phase of $H_k(f)$ satisfies
\begin{equation}\label{eq:H_k_phase_distribution}
\arg[H_k(f)]=\theta_k(f),
\end{equation}
and $|H_k(f)|^2$ follows a lognormal distribution as
\begin{equation}\label{eq:H_k_2_distribution}
\log_e\left[ |H_k(f)|^2\right]\sim\mathcal N\left(\bar H_k(f)+2\bar{E}_k(f),4\bar{D}_k(f)\right).
\end{equation}
\end{theorem}
\begin{proof}
Based on~\eqref{eq:hkf} and \eqref{eq:M_k_rw}, we can obtain
\begin{equation}\label{eq:H_k}
\begin{split}
H_k(f)&=\frac{j2\pi fR_{L,k}\mathbf{o}\mathbf{\bar M}_{k}\mathbf{o}^T}{Z_{u,k}(f)Z_{a}(f)}G_k(f)
\\&=G_k(f)\exp\left[\widetilde H_k(f)\!\right]\exp\left[j\theta_k(f)\right].
\end{split}
\end{equation}
It can be observed from \eqref{eq:H_k} that $\arg\left[H_k(f)\right]=\theta_k(f)$.
Plugging \eqref{eq:G_kf_distribution} into \eqref{eq:H_k}, we can derive the distribution of $|H_k(f)|$ as~\eqref{eq:H_k_module_distribution}. Based on~\eqref{eq:H_k_module_distribution}, we can further obtain the distribution of $|H_k(f)|^2$ as \eqref{eq:H_k_2_distribution}.
Therefore, Theorem 1 follows.
\end{proof}

\textit{Remarks on Theorem 1:} Theorem 1 reveals that if underground mediums are very diverse and irregular, both the amplitude attenuation and the power attenuation of a MI signal follow lognormal distributions. Also, the phase shift is independent of eddy currents. This is because the operating frequencies of MI communications are usually low and the transmission distances between the EAP and TUs are much shorter than wavelengthes, such that the phase shifts due to propagation are omitted in \eqref{eq:G_kf}~\cite{Joint_Channel_Antenna_Modeling}, and therefore not reflected in~\eqref{eq:H_k}.

\section{Performance Analyses}
In this section we analyze the performance of MI based multi-band communication under diverse medium fading channel and derive the SNR distribution, ergodic capacity, average BER, and outage probability for both multiplexing and diversity cases.
\subsection{Performance Analyses for Multiplexing Case}
In multiplexing case, the EAP transmits $N+1$ independent data streams in parallel to the $k$-th TU in $\mathcal T_k$. Detailed performance analysis for the multiplexing case are presented as below.
\subsubsection{SNR distribution}
According to~\eqref{eq:ykfac_n}, we can write the SNR of the $k$-th TU in band $\mathcal{B}_{n}$, denoted by $\gamma_{k,n}$, as follows:
\begin{equation}\label{eq:snr_n}
\begin{split}
\gamma_{k,n}= \frac{\int_{\mathcal{T}_{k}}|H_k(f_{n})|^2|x_{k,n}(t)|^2dt}{\int_{\mathcal{T}_{k}}\mathbb E [|w_{k,n}(t)|^2]dt}
=\frac{|H_k(f_{n})|^2P_{k,n}}{N_0B_n}.
\end{split}
\end{equation}
 Substituting \eqref{eq:H_k_2_distribution} into \eqref{eq:snr_n}, we can obtain the distribution of $\gamma_{k,n}$ as follows:
 \begin{equation}\label{eq:snr_n_distribution}
\begin{split}
\log_e\left[\gamma_{k,n}\right]\sim\mathcal N&\Big( E_{\gamma_{k,n}}, D_{\gamma_{k,n}}\Big),
\end{split}
\end{equation}
where
\begin{equation}
\begin{split}
E_{\gamma_{k,n}}=\bar H_k(f_{k,n})+2\bar{E}_k(f_{k,n})+\log_e\left[\frac{P_{k,n}}{N_0B_n}\right]
\end{split}
\end{equation}
and $D_{\gamma_{k,n}}=4\bar{D}_k(f_{k,n})$.
The corresponding probability density function (PDF) of $\gamma_{k,n}$, denoted by $\hat f_{\gamma_{k,n}}(x)$, can be given by
\begin{equation}\label{eq:PDF_gamma_k(f)}
\hat f_{\gamma_{k,n}}(x)\!=\!\left\{\begin{array}{l}
\!\!\!\!\dfrac{1}{x \sqrt{2 \pi D_{\gamma_{k,n}}}} \!\exp\!\left[\!{-\dfrac{(\log_e x-E_{\gamma_{k,n}})^2}{2D_{\gamma_{k,n}}}}\!\right]\!, x\!>\!0; \\
\!\!\!\!0, x \leq 0.
\end{array}\right.
\end{equation}
The mean and variance of $\gamma_{k,n}$, denoted by $\widetilde E_{\gamma_{k,n}}$ and $\widetilde D_{\gamma_{k,n}}$, can be calculated as
\begin{equation}
\begin{split}
\widetilde E_{\gamma_{k,n}}=\exp\left[ E_{\gamma_{k,n}}+\frac{1}{2} D_{\gamma_{k,n}}\right]
\end{split}
\end{equation}
and
\begin{equation}
\begin{split}
\widetilde D_{\gamma_{k,n}}=\left(\exp\left[ D_{\gamma_{k,n}}\right]-1\right)\exp\left[2 E_{\gamma_{k,n}}+ D_{\gamma_{k,n}}\right],
\end{split}
\end{equation}
respectively.

\textit{Remarks:} Since $\gamma_{k,n}$ is the product of $|H_K(f_n)|^2$ and the constant $P_{k,n}/(B_nN_0)$, it follows a lognormal distribution. Based on this distribution, we will derive closed-form expressions for ergodic capacity, average BER, and network outage probability in multiplexing case in the following.
\subsubsection{Ergodic capacity}
Based on \eqref{eq:PDF_gamma_k(f)}, we can obtain the ergodic capacity corresponding to the $n$-th band of the $k$-th TU (in Nat/s), denoted by $C_{k,n}$, as follows:
\begin{equation}\label{eq:C_kn}
\begin{split}
C_{k,n}=\int_{0}^{+\infty }\!\!\!\!B_{n}\log_e\left(1+x\right)\hat f_{\gamma_{k,n}}(x)dx.
\end{split}
\end{equation}
\begin{figure*}[t]
\begin{equation}\label{eq:Ergodic_Capacity_Log_Normal_Est}
\begin{aligned}
{C}_{k,n} &\approx\frac{B_{n}}{2}\exp \left(-\frac{E_{\gamma_{k,n}}^2}{2 D_{\gamma_{k,n}}}\right) \sum_{i=1}^8 a_i\left[\operatorname{erfcx}\left(\sqrt{\frac{D_{\gamma_{k,n}}}{2}} i+\frac{E_{\gamma_{k,n}}}{\sqrt{2D_{\gamma_{k,n}}}}\right)+\operatorname{erfcx}\left(\sqrt{\frac{D_{\gamma_{k,n}}}{2}} i-\frac{E_{\gamma_{k,n}}}{ \sqrt{2D_{\gamma_{k,n}}}}\right)\right] \\
& +B_{n}\sqrt{\frac{D_{\gamma_{k,n}}}{2 \pi}} \exp \left(-\frac{E_{\gamma_{k,n}}^2}{2 D_{\gamma_{k,n}}}\right)+\frac{B_{n}E_{\gamma_{k,n}}}{2} \operatorname{erfc}\left(-\frac{E_{\gamma_{k,n}}}{ \sqrt{2D_{\gamma_{k,n}}}}\right).
\end{aligned}
\end{equation}
\hrulefill
\end{figure*}Based on the derivations of~\cite{Ergodic_Capacity_Log_Normal}, \eqref{eq:C_kn} can be accurately estimated with~\eqref{eq:Ergodic_Capacity_Log_Normal_Est}, where $a_i$ are constants defined in~\cite{HandBook}, $\operatorname{erfc}(x)= 1-2\int_0^{x}\exp(-t^2)dt/\sqrt{\pi}$ is the complementary error function, and $\operatorname{erfcx}(x)= \exp(x^2) \operatorname{erfc}(x)$ is the scaled complementary error function. The error of~\eqref{eq:Ergodic_Capacity_Log_Normal_Est} can be upper-bounded by $3B_{n}\times10^{-8}$. Both $\operatorname{erfc}(x)$ and $\operatorname{erfcx}(x)$ are built-in Matlab functions and easily calculated. Thus, \eqref{eq:C_kn} can be numerical evaluated.
Then, we can obtain the ergodic capacity of the multi-band communication network in multiplexing case, denoted by $C_{M}$, as $C_{M}=\frac{1}{K+1}\sum_{k=0}^{K}\sum_{n=0}^{N}{C}_{k,n}$.
\subsubsection{Average BER}
Based on~\eqref{eq:PDF_gamma_k(f)}, the average BER corresponding to the transmission in $\mathcal{B}_{n}$ can be given by
\begin{equation}\label{eq:P_ekn}
\begin{split}
P_{e,kn}=\int_{0}^{+\infty}\alpha_{k,n} Q(\sqrt{\beta_{k,n} x})\hat f_{\gamma_{k,n}}(x)dx,
\end{split}
\end{equation}
where $\alpha_{k,n}$ and $\beta_{k,n}$ are the parameters depending on modulation schemes, and $Q(s)\!=\!\int_0^{\pi / 2} \exp \left[-s^2/{2 \sin ^2 x}\right] dx/\pi$~\cite{Wireless_Communications}. Let $\mathcal{M}_{\gamma_{k,n}}(s)=\int_{0}^{+\infty}\hat f_{\gamma_{k,n}}(x)\exp(-sx)dx$ denote the moment generating function (MGF) of $\gamma_{k,n}$. Substituting $Q(s)$ into \eqref{eq:P_ekn}, $P_{e,kn}$ can be re-written with MGF as follows:
\begin{equation}\label{eq:P_ekn_rw}
\begin{split}
P_{e,kn}&\!\!=\!\!\frac{\alpha_{k,n}}{\pi}\! \int_0^{+\infty}\!\int_0^{\pi / 2}\! \exp\left(\frac{-\beta_{k,n}s}{2\sin ^2 x}\right)dx\hat f_{\gamma_{k,n}}(s) d s
\\&\!\!=\!\!\frac{\alpha_{k,n}}{\pi}\! \int_0^{\pi / 2}\!\left[\int_0^{+\infty}\!\!\! \exp\left(\frac{-\beta_{k,n}s}{2\sin ^2 x}\right)\!\hat f_{\gamma_{k,n}}(s) d s\right]\!dx
\\&\!\!=\!\!\frac{\alpha_{k,n}}{\pi} \int_0^{\pi / 2} \mathcal{M}_{\gamma_{k,n}}\left(\frac{\beta_{k,n}}{2\sin ^2 x}\right) d x.
\end{split}
\end{equation}
To the best of our knowledge, there is no general closed-form expression for the lognormal MGF. The MGF of a lognormal random variable can be approximated by a series expansion based on Gauss-Hermite integration. For real $s$, $\mathcal{M}_{\gamma_{k,n}}(s)$ can be approximated as follows~\cite{HandBook}:
\begin{equation}\label{eq:M_kn}
\begin{split}
\mathcal{M}_{\gamma_{k,n}}(s)\!\!=\!\!\sum_{i=1}^{\tilde N} \frac{b_i}{\sqrt{\pi}} \exp\! \left[-s \exp \left({c_i\sqrt{2D_{\gamma_{k,n}}}\!\!+\!\!E_{\gamma_{k,n}}}\right)\!\right],
\end{split}
\end{equation}
where the parameters $b_i$ and $c_i$ can be obtained from existing databases such as numpy. In Section~V we set $\bar N=500$. With $\eqref{eq:M_kn}$, it is convenient to evaluate the average BER given by \eqref{eq:P_ekn_rw} with numerical integration because of the finite integration limits. Then, we can obtain the average BER of the multi-band communication network in multiplexing case, denoted by $P_{e,M}$, as $P_{e,M}={\sum_{k=0}^{K}\sum_{n=0}^{N}P_{e,kn}}/[(N+1)(K+1)]$.
\subsubsection{Outage probability (deep fading probability)}
The outage probability corresponding to the $n$-th band of the $k$-th TU, denoted by $P_{o,kn}$, is defined as the probability that $\gamma_{k,n}$ is less than a threshold $\gamma_{th}$, which is synonymous with the probability of deep fading. $P_{o,kn}$ can be given by
\begin{equation}\label{eq:P_o_kn}
\begin{split}
P_{o,kn}&\!=\!\int_{0}^{\gamma_{th}}\!\!\hat f_{\gamma_{k,n}}(x)dx
\\&\!=\!\hat F_{\gamma_{k,n}}(\gamma_{th})
\!=\!\frac{1}{2}\!\left[\!1\!+\!\operatorname{erf} \left(\frac{\log_e \gamma_{th}-E_{\gamma_{k,n}}}{\sqrt{2D_{\gamma_{k,n}}}}\right)\right],
\end{split}
\end{equation}
where $\hat F_{\gamma_{k,n}}(x)$ is the cumulative distribution function (CDF) of $\gamma_{k,n}$ and $\operatorname{erf}(x)= 2\int_0^{x}\exp(-t^2)dt/\sqrt{\pi}$ is the error function. We denote by $P_{o,M}$ the outage probability of the multi-band network in multiplexing case, which represents the probability that all bands are unavailable for TUs.
Based on \eqref{eq:P_o_kn}, $P_{o,M}$ can be given by $P_{o,M}=\prod_{k=0}^{K} \prod_{n=0}^{N}P_{o,kn}$.

\subsection{Performance Analyses in Diversity Case}
In diversity case, the $N+1$ bands convey the same information. Thus, in diversity case we can write $x_{k,n}(t)$ as $x_{k,n}(t)=x_{k,b}(t)\exp(j2\pi f_n t)$, where $x_{k,b}(t)$ is the base-band signal transmitted to the $k$-th TU. The power of $x_{k,b}(t)$ on a unit resistance, denoted by $P_{k,b}$, is $P_{k,b}=(K+1)\int_{\mathcal T_k}|x_{k,b}(t)|^2dt$.
Let $X_{k,b}(f)$ denote the Fourier transform of $x_{k,b}(t)$. Most energy of $X_{k,b}(f)$ is concentrated within $\mathcal{B}_{b}=\{f|-B_b/2\leq f\leq B_b/2\}$ and satisfies $\int_{\mathcal{B}_{b}}|X_{k,b}(f)|^2df\!\approx\!\int_{\mathcal T_k}|x_{k,b}(t)|^2dt\!=\!P_{k,b}/(K\!+\!1)$. Then, the receive signal corresponding to the $n$-th band of the $k$-th TU in diversity case can be given by
\begin{equation}\label{eq:ykfac_n_rw}
\begin{split}
y_{k,n}(t)=H_k(f_{n})x_{k,b}(t)\exp(j2\pi f_n t)+w_{k,n}(t),t\in\mathcal{T}_{k}.
\end{split}
\end{equation}
Detailed performance analysis for the diversity case are presented as below.
\subsubsection{SNR distribution}
It is assumed that maximum ratio combining (MRC) is used for the $N+1$ bands. Then, we can write the SNR of the $k$-th TU in diversity case, denoted by $\gamma_{k,d}$, as follows:
\begin{equation}\label{eq:snr_d}
\begin{split}
\gamma_{k,d}=\frac{\left(\sum_{n=0}^{N}|H_k(f_{n})|^2\right)P_{k,b}}{B_bN_0}.
\end{split}
\end{equation}
It is shown in Section III that $|H_k(f_{n})|^2$ follows a lognormal distribution. Thus, to investigate the distribution of $\gamma_{k,d}$, it is desired to know the distribution of the sum of lognormal random variables $\sum_{n=0}^{N}|H_k(f_{n})|^2$. However, the general closed-form expression for the PDF of the lognormal sum are not known. It has been recognized that the lognormal sum can be well approximated by a new lognormal~\cite{Lognormal_Sum}. We can approximate $\sum_{n=0}^{N}|H_k(f_{n})|^2$ as a lognormal variable $\exp[\Lambda_k]$, where $\Lambda_k$ follows $\mathcal N( E_{\Lambda_k}, D_{\Lambda_k})$. Then, matching the MGF of $\exp[\Lambda_k]$ and the MGF of $\sum_{n=0}^{N}|H_k(f_{n})|^2$, we have the equation shown as~\eqref{eq:MGF_Matching}. To obtain $E_{\Lambda_k}$ and $D_{\Lambda_k}$, \eqref{eq:MGF_Matching} can be solved numerically using standard functions such as fsolve in Matlab and NSolve in Mathematica with two real and positive values $s_1$ and $s_2$. Then, the PDF of $\gamma_{k,d}$, denoted by $\hat f_{\gamma_{k,d}}(x)$, can be given by
\begin{equation}\label{eq:PDF_gamma_kd}
\hat f_{\gamma_{k,d}}(x)\!=\!\left\{\begin{array}{l}
\!\!\!\!\dfrac{1}{x \sqrt{2 \pi D_{\Lambda_k}}} \!\exp\!\left[\!{-\dfrac{(\log_e x-E_{\gamma_{k,d}})^2}{2D_{\Lambda_k}}}\!\right]\!, x\!>\!0; \\
\!\!\!\!0, x \leq 0,
\end{array}\right.
\end{equation}
where
\begin{equation}
\begin{split}
E_{\gamma_{k,d}}={E}_{\Lambda_k}+\log_e\left[\frac{P_{k,b}}{B_bN_0}\right].
\end{split}
\end{equation}
The mean and variance of $\gamma_{k,d}$, denoted by $\widetilde E_{\gamma_{k,d}}$ and $\widetilde D_{\gamma_{k,d}}$, can be calculated as
\begin{equation}
\begin{split}
\widetilde E_{\gamma_{k,d}}=\exp\left[ E_{\gamma_{k,d}}+\frac{1}{2} D_{\Lambda_k}\right]
\end{split}
\end{equation}
\begin{figure*}[t]
\begin{equation}\label{eq:MGF_Matching}
\begin{split}
\prod_{n=0}^{N}\left\{\sum_{i=1}^{\tilde N} \frac{b_i}{\sqrt{\pi}} \exp\! \left[-s \exp \left({c_i2\sqrt{2\bar{D}_k(f_n)}\!\!+\!\!\bar H_k(f_n)+2\bar{E}_k(f_n)}\right)\!\right]\right\}=\sum_{i=1}^{\tilde N} \frac{b_i}{\sqrt{\pi}} \exp\! \left[-s \exp \left({c_i\sqrt{2D_{\Lambda_k}}\!\!+\!\!E_{\Lambda_k}}\right)\!\right].
\end{split}
\end{equation}
\hrulefill
\begin{equation}\label{eq:Ergodic_Capacity_Log_Normal_Est_Diversity}
\begin{aligned}
{C}_{k,d} &\approx\frac{B_{b}}{2}\exp \left(-\frac{E_{\gamma_{k,d}}^2}{2 D_{\Lambda_k}}\right) \sum_{i=1}^8 a_i\left[\operatorname{erfcx}\left(\sqrt{\frac{D_{\Lambda_k}}{2}} i+\frac{E_{\gamma_{k,d}}}{\sqrt{2D_{\Lambda_k}}}\right)+\operatorname{erfcx}\left(\sqrt{\frac{D_{\Lambda_k}}{2}} i-\frac{E_{\gamma_{k,d}}}{ \sqrt{2D_{\Lambda_k}}}\right)\right] \\
& +B_{b}\sqrt{\frac{D_{\Lambda_k}}{2 \pi}} \exp \left(-\frac{E_{\gamma_{k,d}}^2}{2 D_{\Lambda_k}}\right)+\frac{B_{b}E_{\gamma_{k,d}}}{2} \operatorname{erfc}\left(-\frac{E_{\gamma_{k,d}}}{ \sqrt{2D_{\Lambda_k}}}\right).
\end{aligned}
\end{equation}
\hrulefill
\end{figure*}and
\begin{equation}
\begin{split}
\widetilde D_{\gamma_{k,d}}=\left(\exp\left[ D_{\Lambda_k}\right]-1\right)\exp\left[2 E_{\gamma_{k,d}}+ D_{\Lambda_k}\right],
\end{split}
\end{equation}
respectively.

\textit{Remarks:} With MGF matching, $\sum_{n=0}^{N}|H_k(f_{n})|^2$ is approximated with a lognormal random variable. Since $\gamma_{k,d}$ is the product of $\sum_{n=0}^{N}|H_k(f_{n})|^2$ and the constant $P_{k,d}/(B_bN_0)$, it follows a lognormal distribution. Based on this distribution, we will derive closed-form expressions for ergodic capacity, average BER, and network outage probability in diversity case in the following.
\subsubsection{Ergodic capacity}Following similar steps as Section~III-A, we can derive the ergodic capacity corresponding to the $k$-th TU in diversity case, denoted by $C_{k,d}$ as \eqref{eq:Ergodic_Capacity_Log_Normal_Est_Diversity}.
Then, we can obtain the ergodic capacity of the multi-band communication network in diversity case, denoted by $C_{D}$, as $C_{D}=\frac{1}{K+1}\sum_{k=0}^{K}{C}_{k,d}$.
\subsubsection{Average BER}
The average BER corresponding to the $k$-th TU in diversity case, denoted by $P_{e,k,d}$, can be derived as follows:
\begin{equation}\label{eq:P_ekd_rw}
\begin{split}
P_{e,k,d}=\frac{\alpha_{k,d}}{\pi} \int_0^{\pi / 2} \mathcal{M}_{\gamma_{k,d}}\left(\frac{\beta_{k,d}}{2\sin ^2 x}\right) d x,
\end{split}
\end{equation}
where $\alpha_{k,d}$ and $\beta_{k,d}$ are the parameters depending on the modulation scheme of the $k$-th TU, and
\begin{equation}\label{eq:M_kn_d}
\begin{split}
\mathcal{M}_{\gamma_{k,d}}(s)\!\!=\!\!\sum_{i=1}^{\tilde N} \frac{b_i}{\sqrt{\pi}} \exp\! \left[-s \exp \left({c_i\sqrt{2D_{\Lambda_k}}\!\!+\!\!E_{\gamma_{k,d}}}\right)\!\right].
\end{split}
\end{equation}
Then, we can obtain the average BER of the multi-band communication network in diversity case, denoted by $P_{e,D}$, as $P_{e,D}={\sum_{k=0}^{K}P_{e,k,d}}/(K+1)$.
\subsubsection{Outage probability}
The outage probability corresponding to the $k$-th TU in diversity case can be derived as follows:
\begin{equation}\label{eq:P_o_kd}
\begin{split}
P_{o,k,d}\!=\!\frac{1}{2}\!\left[\!1\!+\!\operatorname{erf} \left(\frac{\log_e \gamma_{th}-E_{\gamma_{k,d}}}{\sqrt{2D_{\Lambda_k}}}\right)\right].
\end{split}
\end{equation}
The outage probability of the multi-band network in diversity case, denoted by $P_{o,D}$, can be given by $P_{o,D}=\prod_{k=0}^{K} P_{o,k,d}$.
\section{Numerical Evaluations}
In this section, numerical evaluations are conducted to show the performances of MuReC based multi-band MI communication under diverse medium fading channel. The permeability of mediums is set to $\mu = 4\pi\cdot10^{-7}~\text{H/m}$. The average transmission distance is set to $20~\text{m}$. The coil radius of the EAP and TUs are set to $60~\text{cm}$ and $20~\text{cm}$, respectively. The number of turns of EAP coils and TU coils are set to $200$ and $50$, respectively. Following the calculations given in~\cite{Zhi_Sun_Underground}, we can obtain $R_{a}=2.2619~\Omega$ and $R_{u,k}=0.1885~\Omega$. The transmit power and noise power spectral density are set to $P_t=7~\text{dBw}$ and $-185~\text{dBw/Hz}$. We set MI channels contain ten types of mediums: soil, water, concrete, wood, air, copper, aluminum, silver, gold, lead, zinc, tin, and titanium.
For poor conductors soil, water, concrete, wood, and air, we calculate their skin depths by~\cite{Digital_MI_Steven,Zhi_Sun_2013}
\begin{equation}\label{eq:delta_ki}
{\delta_{k,i}}(f)=\frac{1}{2 \pi f \sqrt{\frac{\mu \epsilon_i}{2}\left(\sqrt{1+\frac{\sigma_i^2}{(2 \pi f)^2 \epsilon_i^2}}-1\right)}},
\end{equation}
where $\epsilon_i$ and $\sigma_i$ are the permittivity and conductivity of the $i$-th medium. The relative permittivities of soil, water, concrete, wood, air are set to 5, 80, 4, 2, 1, respectively. The conductivities of soil, water, concrete, wood, air are set to $10^{-6}$~S/m, $5\times10^{-3}$~S/m, $10^{-5}$~S/m, $10^{-8}$~S/m, and $3\times10^{-15}$~S/m, respectively. For good conductors copper, aluminum, silver, gold, lead, zinc, tin, and titanium, we calculate their skin depths with the following well-known formula:
\begin{equation}\label{eq:delta_ki_good_Conductor}
{\delta_{k,i}}(f)=\sqrt{\frac{1}{\pi f \mu \sigma_i}}.
\end{equation}
The conductivities of copper, aluminum, silver, gold, lead, zinc, tin, and titanium are set to $5.8\times10^{7}$~S/m, $3.5\times10^{7}$~S/m, $6.3\times10^{7}$~S/m, $4.5\times10^{7}$~S/m, $5\times10^{6}$~S/m, $1.6\times10^{7}$~S/m, $9\times10^{6}$~S/m, $2.3\times10^{6}$~S/m, respectively.
As shown in Section~III, no matter what distribution $\Delta r_{k,i}$ follows, the amplitude attenuation and power attenuation of a MI signal follows a lognormal distribution. In this section we assume $\Delta r_{k,i}$ follows uniform distribution and take it as an example to demonstrate the superior performances of the proposed MuReC based multi-band transmission.

The CDF of the SNR obtained from Monte Carlo simulation and the CDF derived from Section~IV are shown in Fig.~\ref{fig:CDF}, which also shows the consistency between the theoretically derived distribution and the distribution obtained from Monte Carlo simulation. Figs.~\ref{fig:PDF} and~\ref{fig:CDF} verify that the MI channel in diverse and irregular mediums follows a lognormal distribution fading law.
\begin{figure}
\centering
\includegraphics[width=0.5\textwidth]{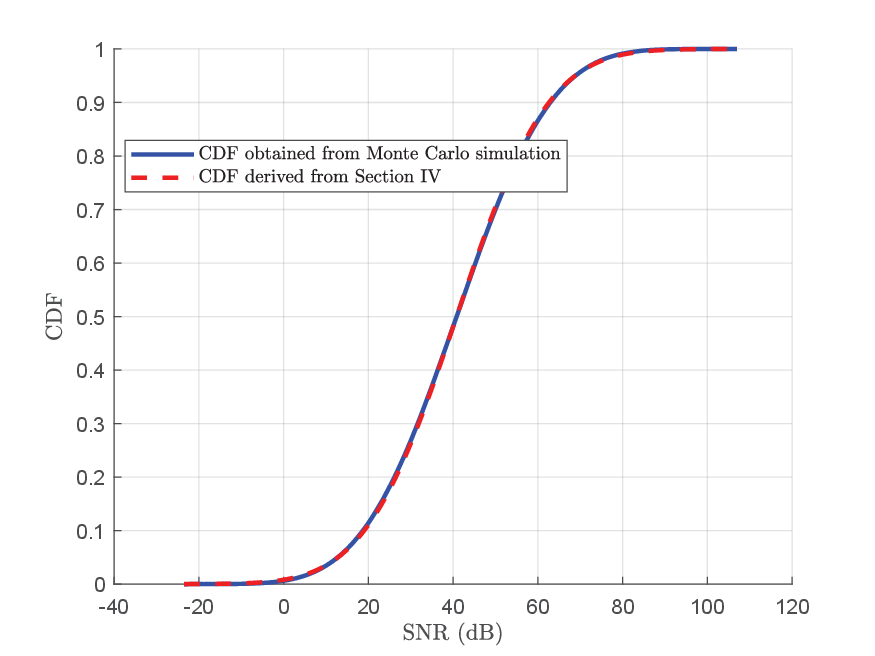}
\caption{The CDF of SNR obtained from Monte Carlo simulation and the CDF of SNR derived from Section~IV.}
\label{fig:CDF}
\end{figure}

Figure~\ref{Ergotic capacity corresponding to the number of subbands under CLT approximation.} compares the ergodic capacity of MuReC based multi-band multiplexing as that of conventional single-band transmission under various numbers of bands, where the number of users is set to $K+1=4$. We set the total bandwidth and the center frequency of single-band transmission to 1~kHz and 50~kHz. The center frequencies of multi-band multiplexing are 40~kHz and 60~kHz. The center frequencies of 4-band multiplexing are 35~kHz, 45~kHz, 55~kHz, and 65~kHz. The center frequencies of 8-band multiplexing are 32.5~kHz, 37.5~kHz, 42.5~kHz, 47.5~kHz, 52.5~kHz, 57.5~kHz, 62.5~kHz, and 67.5~kHz. The bandwidth of each band of multi-band multiplexing is $1/(N+1)$~kHz. As shown in Fig.~\ref{Ergotic capacity corresponding to the number of subbands under CLT approximation.}, MuReC based multi-band multiplexing outperforms single-band transmission in ergodic capacity and the capacity gain increases as the number of bands increases. We can observe that the ergodic capacity can be enhanced by 26\% with 8-band multiplexing, verifying the effectiveness of the proposed scheme in resisting diverse medium fading, which results from diverse and irregular mediums.
\begin{figure}
\centering
\includegraphics[width=0.5\textwidth]{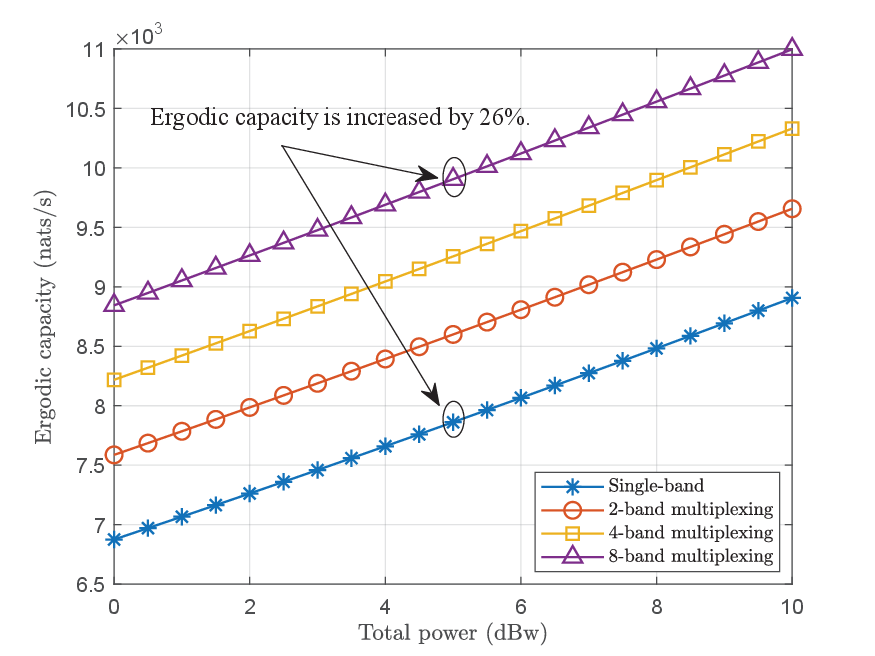}
\caption{Ergodic capacities of conventional single-band transmission and MuReC based multi-band MI multiplexing with various numbers of bands.}
\label{Ergotic capacity corresponding to the number of subbands under CLT
approximation.}
\end{figure}

Figure~\ref{Average BER corresponding to the number of subbands under CLT approximation.} compares the average BER of MuReC based multi-band diversity as that of conventional single-band transmission under various numbers of bands, where the parameter settings are the same as that in Fig.~\ref{Ergotic capacity corresponding to the number of subbands under CLT approximation.}. As shown in Fig.~\ref{Average BER corresponding to the number of subbands under CLT approximation.}, MuReC based multi-band diversity significantly outperforms single-band transmission in average BER and the average BER decreases as the number of bands increases. We can observe that 8-band diversity can reduce the average BER by 4 orders of magnitude, which evaluates the proposed scheme's robustness in complex underground propagation environments.
\begin{figure}
\centering
\includegraphics[width=0.5\textwidth]{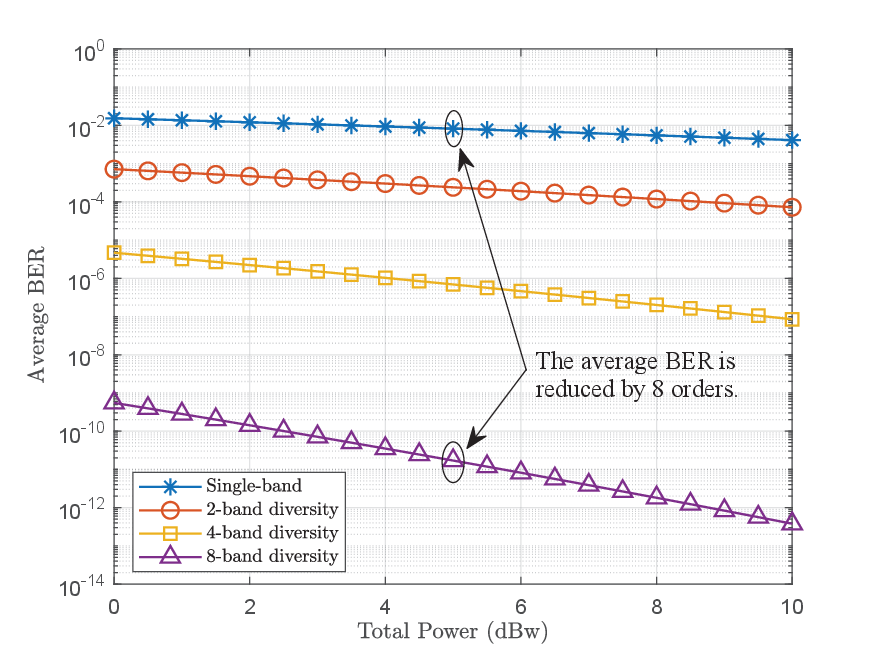}
\caption{Average BERs of conventional single-band transmission and MuReC based multi-band diversity with various numbers of bands.}
\label{Average BER corresponding to the number of subbands under CLT approximation.}
\end{figure}
\begin{figure}
\centering
\includegraphics[width=0.5\textwidth]{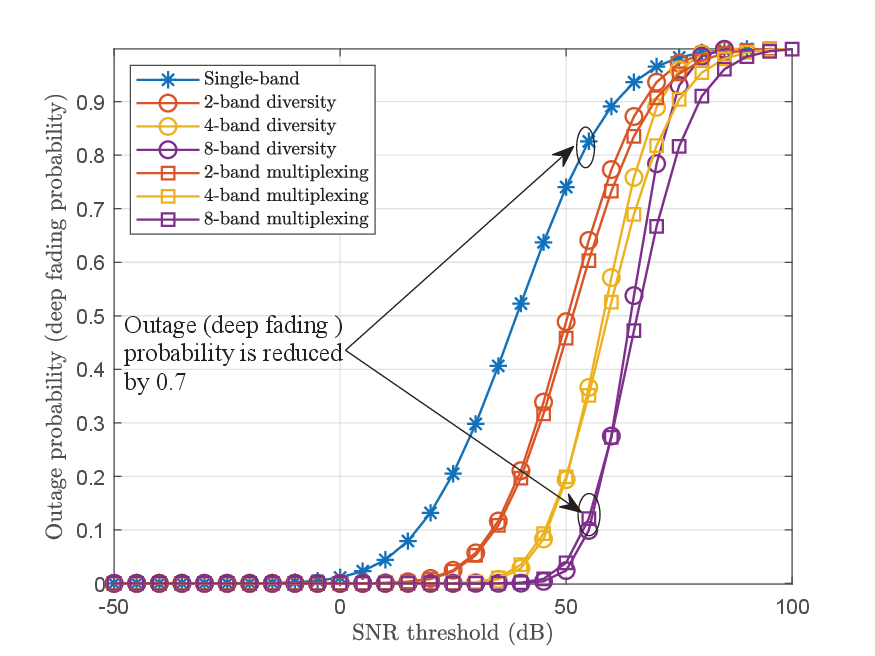}
\caption{Outage probabilities of conventional single-band transmission and MuReC based multi-band diversity with various numbers of bands, where the number of user is $K+1=1$.}
\label{Outage probability corresponding to the number of subband SU.}
\end{figure}

Figure~\ref{Outage probability corresponding to the number of subband SU.} compares the outage probabilities of MuReC based multi-band diversity, MuReC based multi-band multiplexing, and conventional single-band transmission under various numbers of bands, where the number of user is set to $K+1=1$. As shown in Fig.~\ref{Outage probability corresponding to the number of subband SU.}, due to the lack of multiplexing and diversity, the outage probability of single-band transmission increases rapidly as the SNR threshold $\gamma_{th}$ increases. In contrast, MuReC based multi-band schemes effectively utilize the multiple bands to cope with the channel fading, and the outage probabilities decreases as the number of frequencies increases. For the same number of frequencies, the outage probability of MuReC based diversity is slightly smaller than that of MuReC based multiplexing. Setting the number of users to $K+1=4$, Fig.~\ref{Outage probability corresponding to the number of subband MU.} compares the outage probabilities of MuReC based multi-band diversity, MuReC based multi-band multiplexing, and conventional single-band transmission under various numbers of bands. Comparing Figs.~\ref{Outage probability corresponding to the number of subband MU.} and~\ref{Outage probability corresponding to the number of subband SU.}, it can be observed that the outage probability in multi-TUs case is significantly lower than that in single-TU case, which implies that MuReC based multi-band transmission has high reliability in multiple access case.
\begin{figure}
\centering
\includegraphics[width=0.5\textwidth]{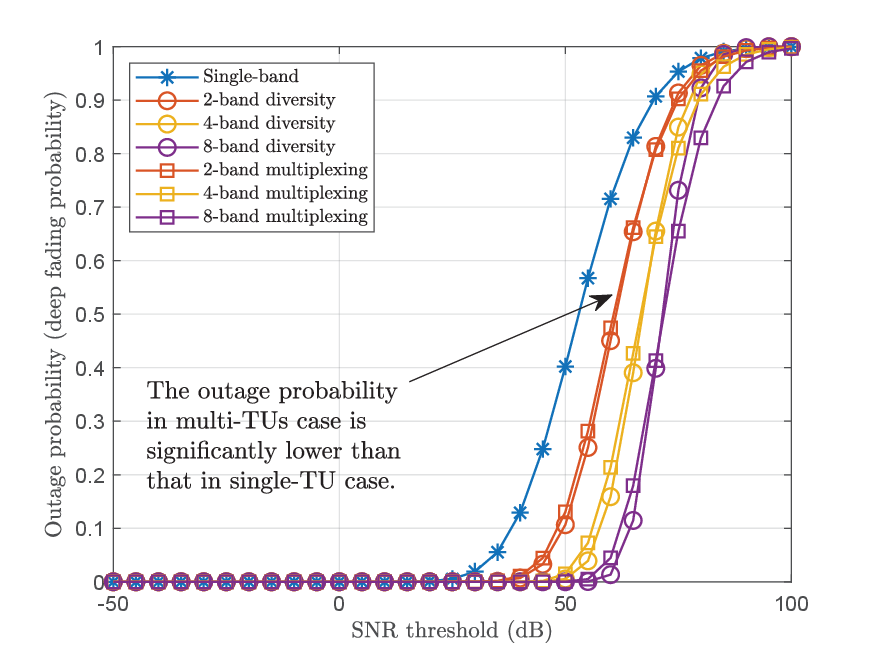}
\caption{Outage probabilities of conventional single-band transmission and MuReC based multi-band diversity with various numbers of bands, where the number of user is $K+1=4$.}
\label{Outage probability corresponding to the number of subband MU.}
\end{figure}
\begin{figure}
\centering
\includegraphics[width=0.5\textwidth]{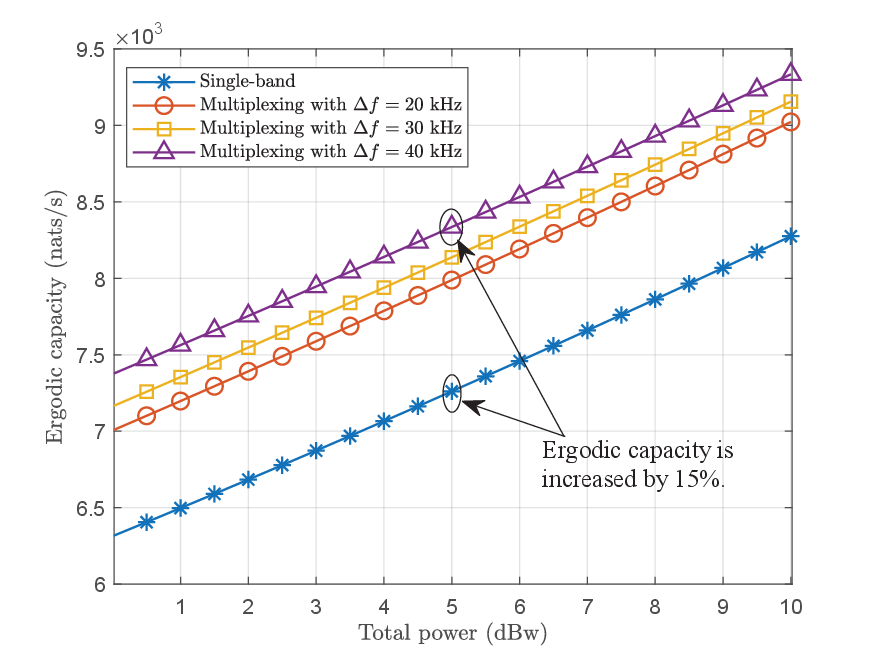}
\caption{Ergodic capacities of conventional single-band transmission and MuReC based multi-band MI multiplexing with various frequency intervals.}
\label{Ergodic capacity corresponding to the Dleta fc under CLT approximation.}
\end{figure}

Figure~\ref{Ergodic capacity corresponding to the Dleta fc under CLT approximation.} shows the ergodic capacity of MuReC based 2-band multiplexing under various frequency intervals and that of conventional single-band transmission, where the number of users is set to $K+1=4$. We set the total bandwidth and the center frequency of single-band transmission to 1~kHz and 50~kHz. The frequency interval of MuReC based multi-band multiplexing is $\Delta f=f_2-f_1$, where $f_2=(50+\Delta f/2)$~kHz and $f_1=(50-\Delta f/2)$~kHz. The bandwidth of each band of 2-band multiplexing is 500~Hz. We can find from Fig.~\ref{Ergodic capacity corresponding to the Dleta fc under CLT approximation.} that the ergodic capacity of MuReC based 2-band multiplexing increases as the frequency interval $\Delta f$ increases. This is because the independence between the two bands increases as the frequency interval $\Delta f$ increases, which reduces the probability of two channels falling into deep fading simultaneously. As compared with single-band transmission, the ergodic capacity can be increased by approximately 15\% using MuReC based 2-band multiplexing with $\Delta f=40$~kHz.

\begin{figure}
\centering
\includegraphics[width=0.5\textwidth]{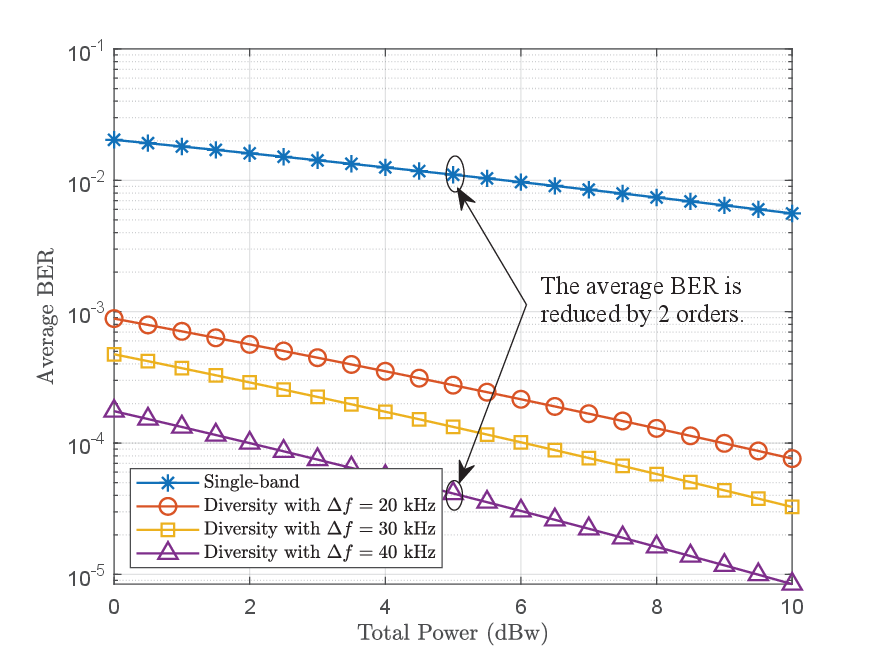}
\caption{Average BERs of conventional single-band transmission and MuReC based multi-band diversity with various frequency intervals.}
\label{Average BER corresponding to the Dleta fc under CLT
approximation.}
\end{figure}

Figure~\ref{Average BER corresponding to the Dleta fc under CLT approximation.} shows the average BER of MuReC based 2-band diversity under various frequency intervals and that of conventional single-band transmission, where the parameter settings are the same as that in Fig.~\ref{Ergodic capacity corresponding to the Dleta fc under CLT approximation.}. It can be found from Fig.~\ref{Average BER corresponding to the Dleta fc under CLT approximation.} that the BER of MuReC based multi-band diversity is significantly lower than that of single-band transmission, and the average BER decreases as $\Delta f$ increases. The average BER can be reduced by 2 orders of magnitude with the frequency interval of $40$~kHz, which implies that the transmission reliability can be enhanced by expanding frequency intervals.
\begin{figure}
\centering
\includegraphics[width=0.5\textwidth]{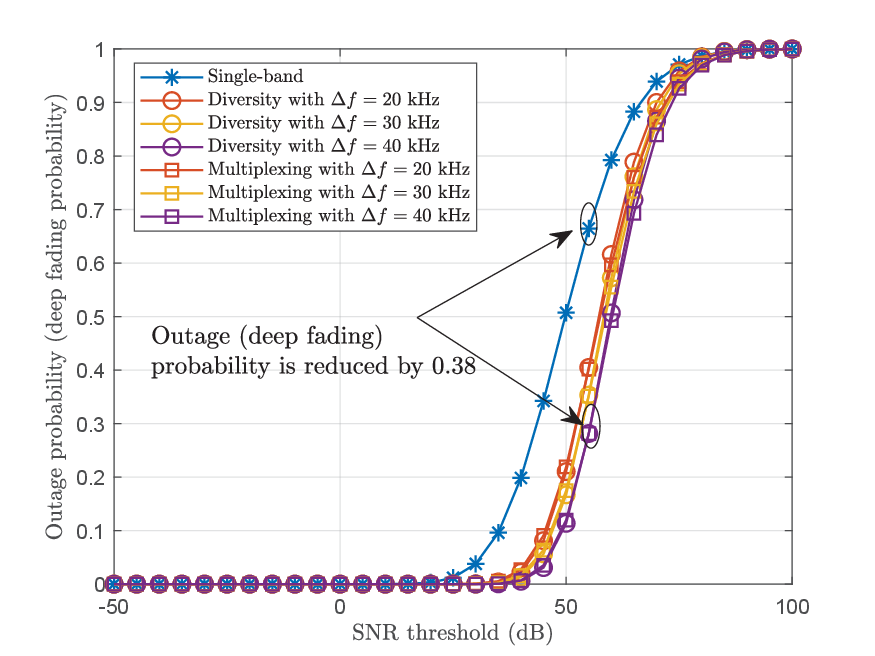}
\caption{Outage probabilities of conventional single-band transmission and MuReC based multi-band diversity with various frequency intervals, where the number of user is $K+1=4$.}
\label{Outage probability corresponding to the Dleta fc under CLT approximation.}
\end{figure}

Figure~\ref{Outage probability corresponding to the Dleta fc under CLT approximation.} shows the outage probabilities of MuReC based multi-band diversity, MuReC based multi-band multiplexing, and conventional single-band transmission under various frequency intervals, where the number of users is set to $K+1=4$. As shown in Fig.~\ref{Outage probability corresponding to the Dleta fc under CLT approximation.}, MuReC based 2-band transmission schemes significantly reduce the outage probability by leveraging frequency intervals and the the outage probability decreases as $\Delta f$ increases. The network outage probability can be reduced by 0.38 with the frequency interval of $40$~kHz, implying that the network reliability in complex post-disaster underground environments can be enhanced by increasing frequency intervals.

\section{Conclusion}
In this paper, we proposed a statistical channel model for MI based underground emergency communications, accounting for the random composition of diverse mediums induced by disasters. The channel fading was analytically and empirically demonstrated to follow a lognormal distribution. To mitigate performance degradation caused by such fading, we introduced MuReC based TD coil to support simultaneous multi-band transmission. By exploiting the low probability of deep fading occurring simultaneously across all bands, the proposed multi-band MI system achieves significant performance improvements over conventional single-band approaches.
Theoretical analysis were conducted to evaluate the system under diverse medium fading. Closed-form expressions for key performance metrics, including SNR probability density functions, ergodic channel capacities, average BERs, and network outage probabilities are derived. Both frequency multiplexing and frequency diversity schemes were investigated. Numerical results confirm that the multi-band MI system offers superior ergodic capacity, lower average BER, and reduced outage probability, making it a robust solution for post-disaster underground environments.
Based on the work in this paper, we will explore adaptive resource allocation, advanced multiplexing and diversity techniques, and experimental validation in practical emergency scenarios to further optimize system performances.
\bibliographystyle{IEEEbib}
\bibliography{ref}
\end{document}